\documentclass[letterpaper, 10pt, conference]{ieeeconf}

\IEEEoverridecommandlockouts

\newtheorem{problem}{Problem}
\newtheorem{theorem}{Theorem}
\newtheorem{corollary}{Corollary}
\newtheorem{lemma}{Lemma}
\newtheorem{remark}{Remark}
\newtheorem{definition}{Definition}
\newtheorem{proposition}{Proposition}
\newtheorem{example}{Example}
\newtheorem{assumption}{Assumption}

\usepackage{amsmath,amsfonts,amssymb,color}
\usepackage{epsfig}
\usepackage{psfrag}
\usepackage{algorithm}
\usepackage{algpseudocode}
\usepackage{array}
\usepackage{epstopdf}
\usepackage{tikz, subcaption, cite}
\usepackage{multicol,multirow}
\usepackage{physics}

\title{\LARGE \bf Policies for Multi-Agency Recovery of Physical Infrastructure After Disasters}
\author{Hemant Gehlot, Shreyas Sundaram, and Satish V. Ukkusuri
\thanks{Hemant Gehlot and Satish V. Ukkusuri are with the Lyles School of Civil Engineering at Purdue University. Email: {\tt \{hgehlot,sukkusur\}@purdue.edu}. Shreyas Sundaram is with the School of Electrical and Computer Engineering at Purdue University. Email: {\tt sundara2@purdue.edu}. This research was supported by National Science Foundation award CMMI 1638311. }
}

\begin{document}
\maketitle
\thispagestyle{empty}
\pagestyle{empty}

\begin{abstract}
We consider a scenario where multiple infrastructure components have been damaged after a disaster and the health value of each component continues to deteriorate if it is not being targeted by a repair agency, until it fails irreversibly. There are multiple agencies that seek to repair the components and there is an authority whose task is to allocate the components to the agencies within a given budget, so that the total number of components that are fully repaired by the agencies is maximized. We characterize the optimal policy for allocation and repair sequencing when the repair rates are sufficiently larger than the deterioration rates. For the case when the deterioration rates are larger than or equal to the repair rates, the rates are homogeneous across the components, and the costs charged by the entities for repair are equal, we characterize a policy for allocation and repair sequencing that permanently repairs at least half the number of components as that by an optimal policy.  
\end{abstract}

\section{Introduction}
Infrastructure systems such as power, transportation, water supply, and sewers, face significant damage after disasters. For example, Hurricane Dorian caused USD 50 million worth of damage to North Carolina roads in 2019 \cite{Doriansw47:online}. Timely recovery of infrastructure systems after disasters is required not only to expedite return of displaced communities, but also to facilitate the movement of emergency teams like ambulances and fire-fighters \cite{arrighi2019preparedness}. 

There are several works in disaster recovery that focus on the repair of infrastructure systems after disasters \cite{coffrin2011strategic,iloglu2020maximal,orabi2010optimizing}. Many of these studies such as \cite{coffrin2011strategic} and \cite{iloglu2020maximal} focus on the recovery of single infrastructure systems like power networks or transportation networks. In addition, most of these studies do not consider accelerated deterioration of infrastructure components after disasters due to processes such as corrosion. For instance, roads experience significant degradation in their strength when they are submerged for prolonged durations after floods  \cite{gaspard2007impact}. Similarly, damage to wastewater pipelines accelerates after the appearance of minor cracks due to disasters \cite{chisolm2012impact}. Due to such deterioration processes, infrastructure components can get damaged to such a level (referred to as the \textit{permanent failure} state) that they require full replacement or repair, which could be very costly. Thus, it is desirable to repair infrastructure components before they reach the permanent failure state. Since infrastructure components face accelerated deterioration after disasters, we assume that the state of a component does not change due to normal deterioration processes once it is fully repaired  (referred to as the \textit{permanent repair} state). 

In this paper, we focus on the problem of maximizing the number of damaged infrastructure components that are permanently repaired after disasters, explicitly focusing on accounting for the dynamics of the deterioration and repair processes. The paper \cite{gehlot2019optimal} focused on this problem considering homogeneous rates of deterioration and repair across the components and \cite{hgehlot_tac} extended the results to the case when the deterioration and repair rates are heterogeneous across the components. However, these studies considered a single agency to repair the components, and did not account for scenarios involving multiple agencies available for recovery after disasters \cite{janssen2010advances}. In addition, there could be a limited budget that the government and other emergency management agencies have in order to pay smaller agencies or teams to repair the infrastructure components \cite{comerio2006estimating,godschalk2009estimating}. Thus, it is important to characterize optimal policies for the allocation of damaged infrastructure to agencies after disasters. 

There are several works that have focused on optimal allocation of resources after disasters \cite{orabi2010optimizing,sun2020post}. However, most of the existing studies do not focus on characterizing optimal policies to allocate infrastructure components that deteriorate with time (if not being repaired). Job scheduling studies \cite{lee2008multi,hsieh1997scheduling} that focus on processing deteriorating jobs over multiple machines  have high-level analogies to our problem. However, job scheduling studies typically do not have a notion of permanent failure, i.e., jobs are processed even if they exceed their due dates, whereas in our problem components cannot be targeted once they permanently fail. Also, jobs are said to be \textit{late} if they are \textit{completed} after their due dates, whereas an agency should \textit{start} targeting a component before it permanently fails.
There are also some high-level analogies of our problem with scheduling of real-time tasks \cite{zhang2009schedulability},
patient triage scheduling problems \cite{argon2008scheduling}, control of thermostatically controlled loads \cite{nilsson2017class,nghiem2011green}, optimal control for persistent monitoring \cite{cassandras2012optimal,lin2014optimal}, resource allocation at a base station to time-dependent queues/flows \cite{eryilmaz2007fair,buche2004control} and machine repair problems \cite{righter1996optimal}. However, most of these studies do not focus on allocation of components under budget constraint and subsequent sequencing for repairing the deteriorating components. Some of these studies \cite{cassandras2012optimal,lin2014optimal,eryilmaz2007fair,buche2004control,zhang2009schedulability,righter1996optimal} also do not consider the notion of permanent failure of components/flows being targeted. 

\subsection{Our contributions}
First, we characterize an optimal policy for allocating components to agencies (and the repair sequencing to be followed by the agencies) when the repair rates are sufficiently larger than the deterioration rates. Secondly, we characterize a 1/2-optimal\footnote{For $\alpha\in (0,1]$, a policy is said to be $\alpha$-optimal if it computes a reward that is at least $\alpha$ times the optimal reward.} policy for allocating components to agencies (and the subsequent repair sequencing) when the deterioration rates are larger than or equal to the repair rates, the deterioration and repair rates are homogeneous across all the components, and the costs charged by the entities for repair are equal; we also prove that the aforementioned policy is optimal when there is a single agency.  

In the next section, we formally describe the recovery problem and after that we characterize recovery policies based on the relative sizes of the deterioration and repair rates. 
\section{Problem Statement}
There are $N(\ge 2)$ infrastructure components (also referred to as \textit{nodes}) represented by the set $\mathcal{V}=\{1,\ldots,N\}$. In the context of disaster recovery, a node can represent a section of road network, or a portion of power network, etc. There are $M(\le N)$ repair agencies (also referred to as \textit{entities}) represented by the set $\mathcal{W}=\{1,\ldots,M\}$. We assume that time progresses in discrete steps, representing the resolution at which the nodes are repaired. We index time-steps by $t\in \mathbb{N}=\{0,1,\ldots\}$. There are two types of control actions. The first type corresponds to the allocation of nodes by an \textit{authority} to entities at time-step $0$. In the context of disaster recovery, an authority could represent a government emergency agency that distributes recovery resources to smaller agencies \cite{comerio2006estimating}.   
For all $h\in \mathcal{W}$, let $\mathcal{U}_h \subseteq \mathcal{V}$ be the set that is allocated to entity $h$. Note that each node $j\in \mathcal{V}$ can be allocated to at most one entity (i.e., $\mathcal{U}_i \cap \mathcal{U}_j = \emptyset$, for all $i \neq j$). The second type of control action corresponds to the repair sequences that are followed by the entities after the allocation of nodes. Let $u_t^h \in \mathcal{U}_h$ be the node that is targeted by entity $h$ at time-step $t$. 

The \textit{health} value of node $j\in \mathcal{V}$ at time-step $t$ is denoted by $v_t^j\in [0,1]$. The initial health value of node $j\in \mathcal{V}$ is denoted by $v_0^j\in (0,1)$. A node $j\in \mathcal{V}$ is said to be \textbf{permanently repaired} at time-step $t$ if $v_{t-1}^j<1$ and $v_t^j=1$. A node $j\in \mathcal{V}$ is said to \textbf{permanently fail} at time-step $t$ if $v_{t-1}^j>0$ and $v_{t}^j=0$. The health value of a node does not change further once it reaches the permanent repair or permanent failure state. The health value of a node $j\in \mathcal{U}_h$ increases by an amount $\Delta_{inc}^{j,h}$ at time-step $t$ if it is being targeted by entity $h$ at time-step $t$ (and node $j$ is not in the permanent repair or permanent failure state at time-step $t$). The health value of a node $j$ decreases by an amount $\Delta_{dec}^{j}$ at time-step $t$ if it is not being targeted by an entity at time-step $t$ (and node $j$ is not in the permanent repair or permanent failure state at time-step $t$). Thus, 
for each node $j\in \mathcal{V}$, the dynamics of its health are given by
\begin{equation} \label{eq:dynamics}
v^j_{t+1} = \begin{cases} 1 &\text{if } v^j_t=1, \\
0 &\text{if } v^j_t=0, \\
\min(1,v^j_t + \Delta_{inc}^{j,h}) & \hbox{if } j\in \mathcal{U}_h, u_t^h = j, \hspace{1mm} \\&\text{and}\hspace{1mm} v^j_t\in(0,1), \\
\max(0, v^j_t-\Delta_{dec}^j) & \hbox{if } u_t^h \neq j, \forall h\in \mathcal{W},   \hspace{1mm} \\&\text{and}\hspace{1mm} v^j_t\in(0,1).\end{cases}
\end{equation}
 
 Each entity $h\in \mathcal{W}$ has a cost per node $c_h\in \mathbb R_{\ge 0}$ that it charges from the authority in order to repair a node. The authority has a constraint on the \textit{budget} $\beta\in \mathbb R_{\ge 0} \cup \{\infty\}$ for paying the entities  (when $\beta=\infty$ we say that there is no budget constraint). Thus, the total cost that is paid to the entities cannot exceed $\beta$ in an allocation, i.e., $\sum_{h=1}^M c_h |\mathcal{U}_h| \le \beta$. 
 \begin{remark}
 Note that although the above budget constraint implies that an entity charges the authority for each node allocated to it regardless of whether the entity permanently repairs that node or not, the policies (i.e., allocation and repair sequencing) that we characterize in this paper ensure that each entity permanently repairs all nodes allocated to it. 
 \end{remark}
 
 We now define the reward function as follows.
 \begin{definition}
 Given a set of initial health values $v_0=\{v_0^1,v_0^2,\ldots,v_0^N\}$, an allocation $\mathcal{U}=\{\mathcal{U}_1,\mathcal{U}_2,\ldots,\mathcal{U}_M\}$ that respects the budget constraint $\beta$, and repair sequences $\bar{u}_{0:\infty}=\{\bar{u}_{0:\infty}^1,\bar{u}_{0:\infty}^2,\ldots,\bar{u}_{0:\infty}^M\}$ for each entity, the reward $J(v_0, \mathcal{U},\bar{u}_{0:\infty})$ is defined as the total number of nodes that get permanently repaired through the repair sequences after the allocation. More formally, $J(v_0, \mathcal{U},\bar{u}_{0:\infty})=|\{j\in \mathcal{V}\hspace{1mm}|\hspace{1mm}\exists \hspace{1mm} t \ge 0 \text{ s.t. } v^j_{t}=1\}|$.
 \end{definition}

Based on the above definitions, the following problem is the focus of this paper.  
\begin{problem} \label{problem}
Given a set $\mathcal{V}$ of $N$ nodes with initial health values $v_0=\{v_0^j\}$, a set $\mathcal{W}$ of $M$ entities with costs $\{c_h\}$, repair rates $\{\Delta_{inc}^{j,h}\}$, deterioration rates $\{\Delta_{dec}^j\}$, and budget $\beta$, find an allocation $\mathcal{U}$ (that respects the budget constraint) and sequences $\bar{u}_{0:\infty}$ that maximize the reward $J(v_0, \mathcal{U},\bar{u}_{0:\infty})$.
\end{problem}

We refer to any given allocation $\mathcal{U}$ and subsequent repair sequences $\bar{u}_{0:\infty}$ as a policy. For our purposes, we now provide the definition of jump.
\begin{definition}
If an entity $h$ starts targeting a node before permanently repairing the node it targeted in the last time-step, then the entity is considered to have jumped at that time-step. Mathematically, if $u_{t-1}^h=j$, $v_{t}^j<1$ and $u_{t}^h \neq j$, then a jump is said to have been made by entity $h$ at time-step $t$. A control sequence that does not contain any jumps is called a \textbf{non-jumping sequence}.
\end{definition}

We now define an $\alpha$-optimal policy as follows.  
\begin{definition}
Let $\{\mathcal{U}^*,\bar{u}^*_{0:\infty}\}$ be an optimal policy (for allocation of nodes and subsequent repair sequences) of Problem \ref{problem}. For $\alpha\in(0,1]$, a policy $\{\mathcal{U},\bar{u}_{0:\infty}\}$ (that respects the budget constraint) is said to be $\alpha$-optimal if $J(v_0, \mathcal{U},\bar{u}_{0:\infty})\ge \alpha J(v_0, \mathcal{U}^*,\bar{u}^*_{0:\infty}), \forall v_0\in [0,1]^N$. 
\end{definition}

We will divide our analysis into two parts: one for the case when the repair rates are sufficiently larger than the deterioration rates and the other for the case when the deterioration rates are larger than or equal to the repair rates.
 
 \section{Policies for $\Delta^{j,h}_{inc} > \Delta_{dec}^j, \forall j\in \mathcal{V}, h\in \mathcal{W}$}
 In this section, we assume that the repair rates are sufficiently larger than the deterioration rates. That is, we make the following assumption.
 \begin{assumption} \label{assum:repair_rate_suffceintly_large}
For all $j \in \{1,\ldots,N\}$ and $h\in \{1,\ldots,M\}$, we assume $\Delta_{inc}^{j,h}>(N-1)\Delta_{dec}^{j}$ and $\Delta_{inc}^{j,h} >\sum_{k\in \{1,\ldots,N\}\setminus j} \Delta_{dec}^{k}$.
 \end{assumption}
 
 We now present the definition of \textit{modified health value}.
 \begin{definition}
 The modified health value of a node at a time-step is its health value minus its deterioration rate. 
 \end{definition}
 
 We start by reviewing some useful results pertaining to repair of nodes by a single entity (from \cite{hgehlot_tac}). 
\subsection{Finding the largest subset of nodes that can be repaired by a single agency}
 We will use the following result (Lemma 4 of \cite{hgehlot_tac}).
 \begin{lemma} \label{lem:necesary_repair_larger}
Suppose there are $N'(\ge 2)$ nodes represented by the set $\mathcal{V}'\subseteq \mathcal{V}$, and Assumption \ref{assum:repair_rate_suffceintly_large} holds. Then, there exists a sequence that allows a given entity to permanently repair $z(\le N')$ nodes if and only if there exists a set $\{i_1,\ldots,i_z\}\subseteq\mathcal{V}'$ such that   
 \begin{equation}
	v_0^{i_j}>(z-j)\Delta_{dec}^{i_j}, \hspace{3mm} \forall j\in \{1,\ldots,z\}. \label{eq:Deltalesslemma_a}
	\end{equation}
 \end{lemma}
 
 Based on Lemma \ref{lem:necesary_repair_larger}, we now present Algorithm \ref{alg:alloc_repair_larger_single entity} that takes a given set $\mathcal{V}'\subseteq \mathcal{V}$, and returns the largest subset $\mathcal{Y}\subseteq \mathcal{V}'$ that can be permanently repaired by a given entity $h$.
 \begin{algorithm}   \caption{Finding the largest set $\mathcal{Y}$ that can be permanently repaired by a given entity}  \label{alg:alloc_repair_larger_single entity}
Suppose Assumption \ref{assum:repair_rate_suffceintly_large} holds and a set $\mathcal{V}'\subseteq \mathcal{V}$ is given. 
Compute $\left \lceil \frac{v_0^j}{\Delta_{dec}^j} \right \rceil$ for each node $j\in \mathcal{V}'$. Set $\mathcal{Y}=\emptyset$ and $z=0$, and repeat the following until the termination criterion is satisfied. 
    \begin{itemize}
        \item  Stop, if there is no node $j\in \mathcal{V}'$ such that $\left \lceil \frac{v_0^j}{\Delta_{dec}^j} \right \rceil>z$. Otherwise, let $j \in \mathcal{V}'$ be the node with the lowest value of $\left \lceil \frac{v_0^j}{\Delta_{dec}^j} \right \rceil$ that satisfies $ \left \lceil \frac{v_0^j}{\Delta_{dec}^j} \right \rceil > z$ among all the nodes in $\mathcal{V}'$. Remove node $j$ from the set $\mathcal{V}'$ and add it to the set $\mathcal{Y}$, and set $z=z+1$.
    \end{itemize}

\end{algorithm}

 We now present the following result (Lemma 6 of \cite{hgehlot_tac}).
 \begin{lemma} \label{lem:largest_number_repair_larger}
 Suppose Assumption \ref{assum:repair_rate_suffceintly_large} holds and a set $\mathcal{V}'\subseteq \mathcal{V}$ is given. Let $\mathcal{Y}$ be the set obtained from Algorithm \ref{alg:alloc_repair_larger_single entity} and let $z=|\mathcal{Y}|$. Then, $z$ is the largest number such that there exists a set  $\mathcal{Y}=\{i_1,\ldots,i_{z}\}\subseteq \mathcal{V}'$ satisfying \eqref{eq:Deltalesslemma_a}. 
 \end{lemma}
 
 We will also use the following result from \cite{hgehlot_tac}.
\begin{lemma} \label{lem:modified_health}
Suppose a set $\mathcal{U}_h\in \mathcal{V}$ of nodes is allocated to an entity $h$ and Assumption \ref{assum:repair_rate_suffceintly_large} holds. In order to maximize the number of nodes that are permanently repaired by $h$, an optimal sequencing policy is that entity $h$ should target the node with the least modified health value at each time step in the set $\mathcal{U}_h$. 
\end{lemma}

Based on these results, we now turn our attention to the problem we are focusing in this paper, namely allocating nodes to \textit{multiple} entities. 

\subsection{Allocating nodes to multiple entities}
We start by presenting a property of the optimal allocation policy.
 \begin{lemma} \label{lem:repair_larger_lowest_cost_larget_set}
 Suppose there are $N(\ge 2)$ nodes, $M(\le N)$ entities, $\beta\in \mathbb R_{\ge 0} \cup \{\infty\}$, and Assumption \ref{assum:repair_rate_suffceintly_large} holds. Then, it is optimal to allocate the largest number of nodes to the lowest cost entity, the second largest number of nodes to the second lowest cost entity, and so on.
 \end{lemma}
 \begin{proof}
 Denote an optimal policy as $A=\{C,D\}$ where $C$ represents the allocation of nodes and $D$ represents the subsequent sequencing policy followed by the entities for repairing the nodes. By Lemma \ref{lem:modified_health}, we can assume without loss of generality that $D$ is the sequencing policy where each entity targets the node with the least modified health in its allocated set at each time-step. Let $y_j$ be the number of nodes that are allocated to entity $j\in\mathcal{W}$ in $C$. Suppose there exists a pair of entities $\{k,l\}\subseteq \mathcal{W}$ such that $c_{k}<c_{l}$ and $y_{k}<y_{l}$. Let $C'$ be an allocation where entity $k$ is allocated the set of nodes that is allocated to entity $l$ in $C$, entity $l$ is allocated the set of nodes that is allocated to $k$ in $C$ and all the remaining entities are allocated the same set of nodes as in $C$ (and the sequencing policy given by Lemma \ref{lem:modified_health} is followed after the allocation). Then, the allocation $C'$ would satisfy the budget constraint because $c_{k}<c_{l}$, $y_{k}<y_{l}$ and allocation $C$ satisfies the budget constraint. Also, the number of nodes that would be permanently repaired by policy $B=\{C',D\}$ would be the same as that by policy $A$ because of Assumption \ref{assum:repair_rate_suffceintly_large} and Lemma \ref{lem:modified_health}. Therefore, one can iteratively apply the aforementioned argument to obtain a policy where the largest number of nodes are allocated to the lowest cost entity, the second largest number of nodes are allocated to the second lowest cost entity, and so on.
Thus, the result follows. 
 \end{proof}
 
 We now present Algorithm \ref{alg:set_alloc_repair_larger} for allocating nodes to entities. The main idea of Algorithm \ref{alg:set_alloc_repair_larger} is to allocate the nodes by going through the entities in the increasing order of their costs using Algorithm \ref{alg:alloc_repair_larger_single entity}. 
\begin{algorithm}   \caption{Allocation of nodes to entities} \label{alg:set_alloc_repair_larger}
Suppose there are $N (\geq 2)$ nodes, $M(\le N)$ entities and Assumption \ref{assum:repair_rate_suffceintly_large} holds. Set $\mathcal{B}=\mathcal{V}$, $\mathcal{C}=\mathcal{W}$ and $\gamma=\beta$.
Repeat the following until the termination criterion is satisfied. 
    \begin{enumerate}
        \item Stop if $\mathcal{B}=\emptyset$, $\mathcal{C}=\emptyset$ or $\gamma<\underline{c}$, where $\underline{c}=\min \{c_h,h\in \mathcal{C}\}$ is the
lowest cost per node among the entities in set $\mathcal{C}$. 
        \item Otherwise, let $s$ be the entity that has the lowest cost $c_s$ in the set $\mathcal{C}$. Remove $s$ from set $\mathcal{C}$. Denote the set that is obtained from Algorithm \ref{alg:alloc_repair_larger_single entity} when $\mathcal{V}'=\mathcal{B}$ as $\mathcal{Y}$. Let $x= \left \lfloor \frac{\gamma}{c_s} \right \rfloor $. If $x\ge |\mathcal{Y}|$, then allocate set $\mathcal{Y}$ to entity $s$, remove set $\mathcal{Y}$ from set $\mathcal{B}$ and let $\gamma=\gamma-c_s|\mathcal{Y}|$. Otherwise, allocate a set $\mathcal{Y}'$ containing $x$ nodes that is an arbitrary subset of set $\mathcal{Y}$, remove set $\mathcal{Y}'$ from set $\mathcal{B}$ and let $\gamma=\gamma-c_s x$. 
    \end{enumerate}
\end{algorithm}

The following is the main result of this section.
 \begin{theorem} \label{thm:repair_larger_multiple_entities_budget}
Suppose there are $N(\ge 2)$ nodes, $M(\le N)$ entities, $\beta\in \mathbb R_{\ge 0} \cup \{\infty\}$, and Assumption \ref{assum:repair_rate_suffceintly_large} holds. 
Then, the allocation provided by Algorithm \ref{alg:set_alloc_repair_larger}, along with the sequencing policy where each entity targets the node with the least modified health value in its allocated set at each time-step, is optimal for Problem \ref{problem}.
\end{theorem}
\begin{proof}
Let $A=\{C,D\}$ be an optimal policy where the allocation of nodes is represented by $C$ and the sequencing policy by $D$. By Lemma \ref{lem:modified_health}, we can assume without loss of generality that $D$ is the sequencing policy where each entity targets the node with the least modified health in its allocated set at each time-step. Without loss of generality, we assume that each entity permanently repairs all nodes allocated to it in policy $A$ (otherwise, the allocated sets can be reduced in size without affecting the optimality of policy $A$).
Suppose that $y$ nodes are permanently repaired by following sequencing policy $D$ after allocation $C$ such that each entity $j\in \mathcal{W}$ permanently repairs $y_j(\ge 0)$ nodes.  
Then, $y_1+\ldots+y_M=y$. 
Without loss of generality, we assume that $\{1,\ldots,M\}$ represents the order of entities in the non-decreasing order of their costs (i.e., for all pairs $\{k,l\}\subseteq \mathcal{W}$, $c_k \le c_l$ if $k\le l$). 

Let $y'$ be the total number of nodes that are permanently repaired when Algorithm \ref{alg:set_alloc_repair_larger} along with sequencing policy $D$ is followed, where each entity $j\in \mathcal{W}$ permanently repairs $y'_j(\ge 0)$ nodes. We now argue that $\sum_{l=1}^M y'_{l}\ge \sum_{l=1}^M y_{l}$.
Consider the case when $M=1$. 
We prove the result through contradiction. Suppose $y'_1<y_1$. Note that there exists a set of nodes $\{i_{1}^1,\ldots,i_{y_1 }^1\}$ that satisfies $ v_0^{i_k^1}>(y_1-k)\Delta_{dec}^{i_k^1}, \quad \forall k \in \{1,\ldots,y_1\}$ by Lemma \ref{lem:necesary_repair_larger} since $y_1$ nodes are permanently repaired by entity $1$ after allocation $C$ (along with sequencing policy $D$). Then, we reach a contradiction because $y'_1=|\mathcal{Y}'_1|$ is the largest number such that there exists a set $\mathcal{Y}'_1\in \mathcal{V}$ that satisfies \eqref{eq:Deltalesslemma_a} when $z$ is replaced with $y'_1$ (by Lemma \ref{lem:largest_number_repair_larger}). Thus, the assumption that $y'_1<y_1$ is false. 

Now consider the case when $M=2$. . 
We again prove by contradiction. Suppose $y'_1+y'_2<y_1+y_2$. Note that $y'_1 \ge y'_2$ because $y'_1=|\mathcal{Y}'_1|$ is the largest number such that there exists a set $\mathcal{Y}'_1\in \mathcal{V}$ that satisfies \eqref{eq:Deltalesslemma_a} when $z$ is replaced by $y'_1$. Also, $y_1\ge y_2$ by Lemma \ref{lem:repair_larger_lowest_cost_larget_set}. 
Note that for each entity $j\in \{1,2\}$, there exists a set $\{i_1^j,\ldots,i_{y_j}^j\}$ that satisfies
$ v_0^{i_k^j}>(y_j-k)\Delta_{dec}^{i_k^j}, \quad \forall k \in \{1,\ldots,y_j\}$ by Lemma \ref{lem:necesary_repair_larger} since $y_j$ nodes are repaired by entity $j$ after allocation $C$ (along with sequencing policy $D$). 
Without loss of generality, we can assume that for all $j\in \{1,2\}, k\in \{1,\ldots,y_j\}$, node $i^j_k$ is the node with the lowest value of $\left\lceil\frac{v_0^{i^j_k}}{\Delta_{dec}^{i^j_k}}\right\rceil$ such that $\left\lceil\frac{v_0^{i^j_k}}{\Delta_{dec}^{i^j_k}}\right\rceil > y_j-k$ in the set $\mathcal{V}\setminus \{i^1_{k+1},\ldots,i^1_{y_1}\}$ (resp. $\mathcal{V}\setminus \{i^1_{1},\ldots,i^1_{y_1},i^2_{k+1},\ldots,i^2_{y_2}\}$) if $j=1$ (resp. $j=2$); otherwise, we could swap $i_k^j$ with the node $l$ that has the lowest value of $\left\lceil\frac{v_0^{l}}{\Delta_{dec}^{l}}\right\rceil$ such that $\left\lceil\frac{v_0^{l}}{\Delta_{dec}^{l}}\right\rceil > y_j-k$ in the set $\mathcal{V}\setminus \{i^1_{k+1},\ldots,i^1_{y_1}\}$ (resp. $\mathcal{V}\setminus \{i^1_{1},\ldots,i^1_{y_1},i^2_{k+1},\ldots,i^2_{y_2}\}$) if $j=1$ (resp. $j=2$) without affecting the optimality of policy $A$. Then, under the aforementioned conditions, Algorithm \ref{alg:set_alloc_repair_larger} would allocate the set of nodes $\{i_1^j,\ldots,i_{y_j}^j\}$ to each entity $j\in \{1,2\}$. However, this leads to a contradiction because Algorithm \ref{alg:set_alloc_repair_larger} should allocate $y'_1+y'_2 (<y_1+y_2)$ nodes.
Therefore, $y'_1+y'_2<y_1+y_2$ does not hold. 

Consider the case when $M=3$. 
We again prove by contradiction. Suppose $y'_1+y'_2+y'_3<y_1+y_2+y_3$. Note that $y'_1 \ge y'_2 \ge y'_3$ from the definitions of $y'_1$, $y'_2$ and $y'_3$. Also, $y_1\ge y_2 \ge y_3$ by Lemma \ref{lem:repair_larger_lowest_cost_larget_set}. 
Note that for each entity $j\in \{1,2,3\}$, there exists a set $\{i_1^j,\ldots,i_{y_j}^j\}$ that satisfies
$ v_0^{i_k^j}>(y_j-k)\Delta_{dec}^{i_k^j}, \quad \forall k \in \{1,\ldots,y_j\}$ by Lemma \ref{lem:necesary_repair_larger} since $y_j$ nodes are repaired by entity $j$ after allocation $C$ (along with sequencing policy $D$). 
Without loss of generality, we can assume that for all $j\in \{1,2,3\}, k\in \{1,\ldots,y_j\}$, node $i^j_k$ is the node with the lowest value of $\left\lceil\frac{v_0^{i^j_k}}{\Delta_{dec}^{i^j_k}}\right\rceil$ such that $\left\lceil\frac{v_0^{i^j_k}}{\Delta_{dec}^{i^j_k}}\right\rceil > y_j-k$ in the set $\mathcal{V}\setminus \{i^1_{k+1},\ldots,i^1_{y_1}\}$ (resp. $\mathcal{V}\setminus \{i^1_{1},\ldots,i^1_{y_1},\ldots,i^{j-1}_{1},\ldots,i^{j-1}_{y_{j-1}},i^j_{k+1},\ldots,i^j_{y_j}\}$) if $j=1$ (resp. $j\ge 2$); otherwise, we could swap $i_k^j$ with the node $l$ that has the lowest value of $\left\lceil\frac{v_0^{l}}{\Delta_{dec}^{l}}\right\rceil$ such that $\left\lceil\frac{v_0^{l}}{\Delta_{dec}^{l}}\right\rceil > y_j-k$ in the set $\mathcal{V}\setminus \{i^1_{k+1},\ldots,i^1_{y_1}\}$ (resp. $\mathcal{V}\setminus \{i^1_{1},\ldots,i^1_{y_1},\ldots,i^{j-1}_{1},\ldots,i^{j-1}_{y_{j-1}},i^j_{k+1},\ldots,i^j_{y_j}\}$) if $j=1$ (resp. $j\ge 2$) without affecting the optimality of policy $A$. Then, under the aforementioned conditions, Algorithm \ref{alg:set_alloc_repair_larger} would allocate the set of nodes $\{i_1^j,\ldots,i_{y_j}^j\}$ to each entity $j\in \{1,2,3\}$. However, this leads to a contradiction because Algorithm \ref{alg:set_alloc_repair_larger} should allocate  $y'_1+y'_2+y'_3 (<y_1+y_2+y_3)$ nodes.
Therefore, the assumption $y'_1+y'_2+y'_3<y_1+y_2+y_3$ is not true. 

Proceeding in the above way we can argue that $\sum_{l=1}^M y'_{l}\ge \sum_{l=1}^M y_{l}$ for all values of $M$. Thus, the number of nodes that are permanently repaired by Algorithm \ref{alg:set_alloc_repair_larger} along with sequencing policy $D$ is not less than that by allocation $C$ along with sequencing policy $D$. Note that allocation provided by Algorithm \ref{alg:set_alloc_repair_larger} satisfies the budget constraint because the largest possible set of size $y'_1$ is allocated to entity $1$, the second largest possible set of size $y'_2$ is allocated to entity $2$ and so on, in Algorithm \ref{alg:set_alloc_repair_larger} and allocation $C$ satisfies the budget constraint. 
\end{proof}
\begin{remark}
Note that Algorithm \ref{alg:alloc_repair_larger_single entity} has polynomial-time complexity as argued in \cite{hgehlot_tac}.  Algorithm~\ref{alg:set_alloc_repair_larger} also has polynomial time complexity because it involves a loop that executes Algorithm \ref{alg:alloc_repair_larger_single entity} at most $M$ times.
\end{remark}

We now give an example to illustrate Algorithm \ref{alg:set_alloc_repair_larger}.
\begin{example}
Consider four nodes $a,b,c$ and $d$ such that $v_0^a=0.05, v_0^b=0.15,v_0^c=0.06$ and $v_0^d=0.07$. Suppose there are two entities $e$ and $f$ such that $c_e=6$ units and $c_f=8$ units. The total budget is $\beta=19$ units. Also, suppose $\Delta_{inc}^{a,e}=\Delta_{inc}^{a,f}=\Delta_{inc}^{b,e}=\Delta_{inc}^{b,f}=\Delta_{inc}^{c,e}=\Delta_{inc}^{c,f}=\Delta_{inc}^{d,e}=\Delta_{inc}^{d,f}=0.4$ and $\Delta_{dec}^{a}=\Delta_{dec}^{b}=\Delta_{dec}^{c}=\Delta_{dec}^{d}=0.1$.  If Algorithm \ref{alg:set_alloc_repair_larger} is followed for the allocation then entity $e$ (which is the least costly among the two entities) is first allocated the nodes $a$ and $b$ as node $a$ is the node $j\in \{a,b,c,d\}$ with the lowest value of $\left \lceil \frac{v_0^j}{\Delta_{dec}^j} \right \rceil$ such that $\left \lceil \frac{v_0^j}{\Delta_{dec}^j} \right \rceil>0$, $b$ is the only node $j\in \{b,c,d\}$ such that $\left \lceil \frac{v_0^j}{\Delta_{dec}^j} \right \rceil>1$ and there is no node $j\in \{c,d\}$ such that $\left \lceil \frac{v_0^j}{\Delta_{dec}^j} \right \rceil>2$. Note that $\gamma=19-6-6=7$ after allocating the nodes to entity $e$. Thus, it is not possible to allocate any nodes to entity $f$ as $\gamma=7<8=c_{F}$. After the allocation of nodes $a$ and $b$ to entity $e$, the sequence of targeting the node with the least modified health value at each time-step is followed by entity $e$, which permanently repairs both nodes. 
\end{example}

\section{Policies for $\Delta_{dec}^j\ge \Delta^{j,h}_{inc}, \forall j\in \mathcal{V}, h\in \mathcal{W}$} \label{sec:deter_larger}
In this section, we analyze the case when the deterioration rates are larger (but not necessarily significantly larger) than the repair rates. We will use the following assumption in this section.
\begin{assumption} \label{assum:partway-step}
Suppose for all $j\in\{1,\ldots,N\}$ and $h\in \{1,\ldots,M\}$,  $\Delta_{inc}^{j,h}=\Delta_{inc}^h$,  $\Delta_{dec}^j=\Delta_{dec}$, $\Delta_{dec}\ge \Delta_{inc}^h$ and $c_h=c'$. Also, for each entity $h\in\{1,\ldots,M\}$, suppose there exists a positive integer $n_h$ such that $\Delta_{dec} = n_h\Delta_{inc}^h$. Also, for each node $j\in\{1,\ldots,N\}$ and entity $h\in\{1,\ldots,M\}$, suppose there exists a positive integer $m_j^h$ such that $1-v^{j}_0 =m_j^h \Delta_{inc}^h$.
\end{assumption}

The above assumption ensures that no node gets permanently repaired partway through a time-step. 

We will use the following result for repair of nodes by a single entity from \cite{gehlot2019optimal}.
\begin{lemma} \label{lem:healthiest_node}
Suppose a set $\mathcal{U}_h$ is allocated to an entity $h$ and Assumption \ref{assum:partway-step} holds. In order to maximize the number of nodes that are permanently repaired by $h$, an optimal sequencing policy is that entity $h$ should target the healthiest node in $\mathcal{U}_h$ at each time step.\footnote{ Equivalently, the optimal sequence is the non-jumping sequence that targets the nodes in decreasing order of their initial health values in the set $\mathcal{U}_h$.}
\end{lemma}

We now start with the following result.
\begin{lemma} \label{lem:one_node_less_non-jumping}
Let there be $N(\geq2)$ nodes, $M(\le N)$ entities, $\beta=\infty$, and suppose Assumption \ref{assum:partway-step} holds. 
Consider a policy $A=\{C,D\}$ where $C$ is the allocation policy and $D$ is a non-jumping sequencing policy where each entity targets a set of nodes in decreasing order of their initial health. 
Suppose $k$ is the $p$th healthiest node (where $1\le p\le M$) at $t=0$. Also, suppose $k$ is not targeted by an entity at $t=0$ but all the top $p-1$ healthiest nodes are targeted at $t=0$ in $D$. Let $a$ be an entity that does not target a node that belongs to the set of top $p-1$ healthiest nodes at $t=0$.
Suppose entity $a$ targets nodes $\{i_1,\ldots,i_f\}$ in sequencing policy $D$ and permanently repairs all of them. Let $x$ be the number of nodes that are permanently repaired in policy $A$.
Consider another policy $B=\{C',E\}$ where allocation $C'$ is the same as $C$ except that $k$ is allocated to entity $a$, and entity $a$ targets node $k$ at $t=0$ and follows the sequence $\{i_1,\ldots,i_{f-1}\}$ afterwards in $E$; all the other entities target the remaining nodes in the same order as that in $D$. Then, at least $x-1$ nodes are permanently repaired in $B$.
\end{lemma}
\begin{proof}
    Note that the allocation $C'$ satisfies the budget constraint since $\beta=\infty$. We first focus on entity $a$. Note that $k\notin \{i_1,\ldots,i_f\}$ because $v_0^k>v_0^{i_1} $ and $a$ targets the nodes $\{i_1,\ldots,i_f\}$ in the decreasing order of their initial health in sequencing policy $D$. 
    Also, the nodes $k,i_1,\ldots,i_{f-1}$ are targeted at an earlier time in sequencing policy $E$ by entity $a$ as compared to the nodes $i_1,\ldots,i_f$ in $D$ since $v_0^k>\max\{v_0^{i_1},\ldots,v_0^{i_f}\}$ and $v_0^{i_1}\ge \ldots \ge v_0^{i_f}$ (as each entity targets a set of the nodes in the decreasing order of their initial health).
    Thus, entity $a$ permanently repairs the nodes $k,i_1,\ldots,i_{f-1}$ in $E$.  

    We now compare the number of nodes that are permanently repaired by the entities apart from entity $a$ in the sequencing policies $D$ and $E$.
    Note that either node $k$ is not permanently repaired in $D$ or node $k$ is permanently repaired by an entity $b(\neq a)$ in $D$ (note that entities $a$ and $b$ cannot be the same because entity $a$ does not target node $k$ in $D$ as mentioned before). Consider the case when no entity permanently repairs node $k$ in $D$. Then, all the entities other than entity $a$ would permanently repair the same number of nodes in $E$ as that in $D$. We now consider the case when node $k$ is permanently repaired by entity $b$ in $D$. Suppose entity $b$ permanently repairs $g$ nodes in $D$. Denote the order that is followed by entity $b$ for targeting the nodes in $D$ as $\{i'_1,i'_2,\ldots,i'_{g}\}$, where $i'_j=k$ such that $1\le j\le g$. Since entity $b$ follows the same order for targeting the nodes in $E$ as that followed in $D$, at least $g-1$ nodes (i.e., the nodes $\{i'_1,\ldots,i'_{j-1},i'_{j+1},\ldots, i'_{g}\}$) would be permanently repaired by $b$ in $E$ as these nodes would start to get targeted at an earlier or same time-step in comparison to $D$. Note that all the remaining entities (i.e., entities apart from $a$ and $b$) would permanently repair the same number of nodes in both $D$ and $E$. Thus, the number of nodes that are permanently repaired in policy $B=\{C',E\}$ is at least equal to $x-1$. 
\end{proof}

We now present the main result of this section that provides an \textit{online} policy where the nodes are sequentially allocated to entities and the entities permanently repair their currently allocated nodes before they are allocated new nodes.
\begin{theorem} \label{thm:multiple_entities_budget_dec_larger_2-approx}
Suppose there are $N(\ge 2)$ nodes, $M(\le N)$ entities, $\beta\in \mathbb R_{\ge 0} \cup \{\infty\}$, and Assumption \ref{assum:partway-step} holds. 
Consider the \textit{online} policy where at each time-step the healthiest node that is currently not being targeted is allocated to an entity that is currently not repairing any node, until there are no more nodes to allocate or the budget runs out.
Then, the aforementioned policy is 1/2-optimal for Problem \ref{problem}.
\end{theorem}
\begin{proof}
Let  $A=\{C,D\}$ be an optimal policy where the allocation of nodes is given by $C$ and the subsequent sequencing policy by $D$. Without loss of generality, we can assume that $D$ is the policy where each entity targets the  healthiest node in its allocated set at each time-step (or, the non-jumping sequence of targeting the allocated nodes in the decreasing order of their initial health) by Lemma \ref{lem:healthiest_node}.
We now argue that the online policy would permanently repair at least half the number of nodes as those by policy $A$ because of the following argument. 

 Let $\mathcal{X}_t\subseteq \mathcal{W}$ contain all the entities that are unassigned at time-step $t$ (i.e., those entities that are not currently busy in targeting previously allocated nodes at time-step $t$) in sequencing policy $D$. Let $\mathcal{Y}_t \subseteq \mathcal{V}$ be the set of all the nodes that have health values in the interval $(0,1)$ at time-step $t$ and have not been targeted before time-step $t$ in sequencing policy $D$. Let $l_t=\min\{|\mathcal{X}_t|,|\mathcal{Y}_t|\}$. We now go to the first time-step $T$ in sequencing policy $D$ at which $l_T \ge 1$ and at least one of the top $l_T$ healthiest nodes is not targeted at time-step $T$.
 Suppose $k$ is the $p$th healthiest node in the set $\mathcal{Y}_T$ (where $1\le p\le l_T$) such that $k$ is not targeted by an entity in $\mathcal{X}_T$ but the $p-1$ healthiest nodes in the set $\mathcal{Y}_T$ are targeted by the entities in $\mathcal{X}_T$. Let $a$ be an entity in the set $\mathcal{X}_T$ that does not target one of the $p-1$ healthiest nodes at time-step $T$. Suppose entity $a$ permanently repairs $f$ nodes from time-step $T$ onwards in $D$. Let $E$ be a sequencing policy which is the same as $D$ from time-step $0$ to $T-1$ but the portion of $E$ from time-step $T$ onwards is such that node $k$ is targeted by entity $a$ at time-step $T$, entity $a$ targets the first $f-1$ nodes that it targeted in $D$ after targeting node $k$ and the remaining entities target the nodes in the same order as in $D$.
 Then, at least $x-1$ nodes would be permanently repaired in sequencing policy $E$ by Lemma \ref{lem:one_node_less_non-jumping}.

 We iteratively repeat the above procedure so that at each time-step $t$ the healthiest node that has not been targeted before is allocated to an entity that is available at time-step $t$. Note that since the number of nodes that are permanently repaired in the given sequencing policy either decreases or remains the same in each iteration of this procedure, the total cost of the permanently repaired nodes does not increase during this procedure as the costs are homogeneous across all the entities by Assumption \ref{assum:partway-step}. Thus, the final allocation satisfies the budget constraint as allocation $C$ is a feasible allocation.  
 Therefore, the aforementioned online policy is 1/2-optimal because 1) at each iteration of this iterative procedure, we move at least one node across the given sequencing policy and the number of nodes that are permanently repaired reduces by at most one, and 2) in the last iteration of this procedure when there is only one node, there is no decrease in the number of nodes that are permanently repaired because if the last node in the given sequencing policy can be permanently repaired then a healthier node can also be permanently repaired as once a node starts to get targeted in a non-jumping sequencing policy it is always permanently repaired.
\end{proof}

\begin{remark}
Note that although the policy provided in Theorem \ref{thm:multiple_entities_budget_dec_larger_2-approx} is \textit{online}, it can also be used as an \textit{offline} policy to first find an allocation of nodes to entities and then the sequencing policy where each entity targets the allocated nodes in the decreasing order of their initial health is followed.
\end{remark}

We now provide an example to illustrate the online policy.  
\begin{example}
Consider four nodes $a,b,c$ and $d$ such that $v_0^a=0.9, v_0^b=0.8,v_0^c=0.6$ and $v_0^d=0.5$. Suppose there are two entities $e$ and $f$ such that $c_e=c_f=c'=6$ units. The total budget is $\beta=23$ units. Also, suppose $\Delta_{inc}^{e}=\Delta_{inc}^{f}=0.1$ and $\Delta_{dec}=0.2$. Then, the nodes $a$ and $b$ are allocated to the entities $e$ and $f$, respectively, (note that $a$ and $b$ could have also been allocated to $f$ and $e$, respectively) at time-step $t=0$ when the online policy is followed. Note that the budget that is available after the allocation of nodes at $t=0$ is equal to $\gamma=23-6-6=11$. Then, the time after time-step $0$ at which the healthiest node gets permanently repaired is $t=1$ when node $a$ gets permanently repaired, and entity $e$ becomes available and thus can be allocated another node. Thus, node $c$ is allocated to entity $e$ and the budget that is available after this allocation is $\gamma=11-6=5$. Since $\gamma=5<6=c'$, it is not possible to allocate node $d$ to any entity. Thus, nodes $a$ and $c$ are allocated to $e$ and node $b$ is allocated to $f$. 
\end{example}

Note that although we assumed that the costs are homogeneous across the entities in Assumption \ref{assum:partway-step}, the problem of finding the optimal policy under this assumption is non-trivial as shown in the following examples.
\begin{example}
Consider three nodes $a, b$ and $c$ such that $v_0^a=0.9, v_0^b=0.8$ and $v_0^c=0.2$. Suppose there are two entities $d$ and $e$ such that $c'=6$ units. The total budget is $\beta=25$ units. Also, suppose $\Delta_{inc}^{d}=\Delta_{inc}^{e}=0.1$ and $\Delta_{dec}=0.2$.
If the nodes are allocated based on the online policy then the nodes $a$ and $b$ are allocated to entities $d$ and $e$, respectively, (note that $a$ and $b$ could have also been allocated to $e$ and $d$, respectively) at $t=0$. Note that by the time node $a$ (i.e., the healthiest node at $t=0$) is permanently repaired, node $c$ permanently fails and thus is not allocated to any entity. Therefore, the nodes $a$ and $b$ are permanently repaired from the online policy. However, if the nodes $a$ and $b$ are allocated to entity $d$, and node $c$ is allocated to entity $e$, then the aforementioned allocation satisfies the budget constraint and it is possible to repair all the nodes by following the sequence where each entity targets the healthiest node in the allocated set at each time-step. Thus, the online policy, in general, is not optimal. Note that although the aforementioned policy is not optimal in this example, it is indeed 1/2-optimal as proved in Theorem \ref{thm:multiple_entities_budget_dec_larger_2-approx}.
\end{example}

Note that the policy of giving the largest subset of nodes that can be repaired to one entity, the second largest subset of nodes that can be repaired to the second entity and so on, is optimal in the above example but it may not be optimal under Assumption \ref{assum:partway-step} as shown next (note that this policy was indeed optimal in the previous section). 
\begin{example}
Consider four nodes $a, b$, $c$ and $d$ such that $v_0^a=0.9, v_0^b=0.8$, $v_0^c=0.4$ and $v_0^d=0.3$. Suppose there are two entities $e$ and $f$ such that $c'=6$ units. The total budget is $\beta=25$ units. Also, suppose $\Delta_{inc}^{e}=\Delta_{inc}^{f}=\Delta_{dec}=0.1$. If an entity (say entity $e$) is to be allocated the largest set of nodes that it can permanently repair, then $e$ would be allocated nodes $a$ and $b$. Then, entity $f$ would be allocated node $c$ as it can only repair one node out of the remaining nodes $c$ and $d$. Thus, node $d$ would not be allocated to any entity in the aforementioned allocation policy and thus only nodes $a, b$ and $c$ are permanently repaired. However, if the proposed online policy is used then nodes $a$ and $b$ are allocated to the two entities (in any order) at $t=0$. After permanently repairing node $a$ (resp. $b$) it is possible to allocate node $c$ (resp. $d$) to the entity that becomes available. Thus, all the nodes are permanently repaired by the online policy in this example.
\end{example}

Although the above examples show that finding the optimal policy in general is non-trivial, the proposed online policy is optimal when there is a single entity as discussed next.
\begin{proposition} \label{prop:single_entity_deter_larger}
Suppose there are $N(\ge 2)$ nodes, $M=1$, $\beta\in \mathbb R_{\ge 0} \cup \{\infty\}$, and Assumption \ref{assum:partway-step} holds. 
Consider the \textit{online} policy where at each time-step the healthiest node that has not been targeted before is allocated to the entity if the entity is currently not repairing any node, until there are no more nodes to allocate or the budget runs out.
Then, the aforementioned policy is optimal for Problem \ref{problem}. 
\end{proposition}
\begin{proof}
The proof of this result starts similarly as the proof of Theorem \ref{thm:multiple_entities_budget_dec_larger_2-approx} by considering an optimal policy $A=\{C,D\}$ where the nodes are targeted in the decreasing order of their initial health value in sequencing policy $D$. Note that since there is only one entity, there will be no time-step $T$ in $D$ where the healthiest node is not targeted by the entity and thus there will be no reduction in the length of $D$ due to the iterative procedure as in Theorem \ref{thm:multiple_entities_budget_dec_larger_2-approx}. Since the allocation $C$ satisfies the budget constraint, the result follows. 
\end{proof}

We now give an example to show that the online policy is not 1/2-optimal when the deterioration and repair rates are heterogeneous across the nodes.  
\begin{example} \label{exmp:2-approx_heterogeneous}
Consider five nodes $a$, $b$, $c$, $d$ and $e$ such that $v_0^a=0.8, v_0^b=0.8,v_0^c=0.6,v_0^d=0.6$ and $v_0^e=0.6$. Suppose there are two entities $f$ and $g$ such that $c'=1$ unit. The total budget is $\beta=6$ units. The deterioration rates are given by  $\Delta_{dec}^{a,f}=\Delta_{dec}^{a,g}=\Delta_{dec}^{b,f}=\Delta_{dec}^{b,g}=0.05, \Delta_{dec}^{c,f}=\Delta_{dec}^{c,g}=\Delta_{dec}^{d,f}=\Delta_{dec}^{d,g}=0.2 $ and $\Delta_{dec}^{e,f}=\Delta_{dec}^{e,g}=0.6$. The repair rates are given by  $\Delta_{inc}^{a,f}=\Delta_{inc}^{a,g}=\Delta_{inc}^{b,f}=\Delta_{inc}^{b,g}=0.05,\Delta_{inc}^{c,f}=\Delta_{inc}^{c,g}=\Delta_{inc}^{d,f}=\Delta_{inc}^{d,g}=0.2$ and $\Delta_{inc}^{e,f}=\Delta_{inc}^{e,g}=0.4$. If the online policy is followed then nodes $a$ and $b$ are allocated to entities $f$ and $g$ (in any order) at $t=0$ but the remaining nodes permanently fail by the time they are reached (as shown in Table \ref{table_example_2-approx-algo}); so only nodes $a$ and $b$ are allocated and are permanently repaired. However, consider the allocation where the nodes $a$, $c$ and $e$ are allocated to entity $f$ and nodes $b$ and $d$ are allocated to entity $g$ (note that this allocation satisfies the budget constraint). After the allocation, entity $f$ permanently repairs all the three allocated nodes when it first targets node $e$, then node $c$ and finally node $a$ (as shown in Table \ref{table_example_first_entity}). Also, entity $g$ permanently repairs the remaining two nodes by first targeting node $d$ and then node $b$ (as shown in Table \ref{table_example_second_entity}). Since all the five nodes are permanently repaired after the aforementioned allocation in comparison to two nodes that are permanently repaired by the online policy, the latter policy is not 1/2-optimal when the deterioration and repair rates are heterogeneous across the nodes.
\end{example}
\begin{table}[h]
	\caption{Health progression in Example \ref{exmp:2-approx_heterogeneous} when the online policy is followed.}
	\label{table_example_2-approx-algo}
	\begin{center}
		\begin{tabular}{|cccccc|}
			\hline
			Time-step $(t)$&$v_t^{a}$& $v_t^{b}$ & $v_t^{c}$ & $v_t^{d}$ & $v_t^{e}$\\ \hline
			0&0.8& 0.8 &0.6& 0.6 &0.6\\
			4&1& 1 &0& 0 &0\\ \hline
		\end{tabular}
	\end{center}
\end{table}
\begin{table}[h]
	\caption{Health progression in Example \ref{exmp:2-approx_heterogeneous} when nodes $a$, $c$ and $e$ are targeted by entity $f$.}
	\label{table_example_first_entity}
	\begin{center}
		\begin{tabular}{|cccc|}
			\hline
			Time-step $(t)$&$v_t^{a}$& $v_t^{c}$ & $v_t^{e}$ \\ \hline
			0&0.8& 0.6 &0.6\\
			1&0.75& 0.4&1 \\
			4&0.6& 1 &1\\ 
			12&1& 1 &1\\ \hline
		\end{tabular}
	\end{center}
\end{table}
\begin{table}[h]
	\caption{Health progression in Example \ref{exmp:2-approx_heterogeneous} when nodes $b$ and $d$ are targeted by entity $g$.}
	\label{table_example_second_entity}
	\begin{center}
		\begin{tabular}{|ccc|}
			\hline
			Time-step $(t)$&$v_t^{b}$& $v_t^{d} $ \\ \hline
			0&0.8& 0.6 \\
			2&0.7& 1 \\
			8&1& 1 \\ \hline
		\end{tabular}
	\end{center}
\end{table}

We now provide an example to show that the online policy is not 1/2-optimal when the costs are heterogeneous across the entities.
\begin{example}
Consider five nodes $a,b,c,d$ and $e$ such that $v_0^a=v_0^b=v_0^c=v_0^d=v_0^e=0.95$. Suppose there are two entities $f$ and $g$ such that $c_f=1$ and $c_g=5$. Also, the total budget is $\beta=6$ units. Suppose $\Delta_{inc}^f=\Delta_{inc}^g=\Delta_{dec}=0.1$. Then, two nodes get allocated to the entities at $t=0$ in the online policy and the remaining budget after the allocation of these nodes is $\gamma=6-1-5=0$. Thus, it is not possible to allocate any more nodes to the entities and therefore two nodes are permanently repaired by the online policy. However, if all the nodes are allocated to entity $f$ (note that this allocation satisfies the budget constraint), then it permanently repairs all the nodes. Since five nodes are permanently repaired by the aforementioned policy in comparison to two nodes by the online policy, the latter policy is not 1/2-optimal when the costs are heterogeneous across the entities. 
\end{example}
\begin{remark}
Note that we do not provide the results for the case when $\Delta_{dec}^j < \Delta_{inc}^{j,h}<(N-1)\Delta_{dec}^j, \forall j \in \mathcal{V}, h\in \mathcal{W}$. The paper \cite{hgehlot_tac} gave some examples to show that when there is a single entity and there is no budget constraint, sequencing policies such as targeting the healthiest node at each time-step or targeting the node with the least modified health value at each time-step need not be optimal for this case. Thus, the analysis of this case remains open for future research.
\end{remark}

\section{Summary}
We characterized policies for allocation and repair sequencing for recovering damaged components after disasters when there are multiple entities available for repair and there is a budget constraint. There can be several future extensions of this work: considering stochasticity in the deterioration and repair rates, introducing interdependencies between the components, characterizing an optimal policy for case in Section \ref{sec:deter_larger} and considering time-dependent deterioration rates are some of those.

\bibliographystyle{IEEEtran}
\bibliography{refs}

\begin{thebibliography}{10}
\providecommand{\url}[1]{#1}
\csname url@samestyle\endcsname
\providecommand{\newblock}{\relax}
\providecommand{\bibinfo}[2]{#2}
\providecommand{\BIBentrySTDinterwordspacing}{\spaceskip=0pt\relax}
\providecommand{\BIBentryALTinterwordstretchfactor}{4}
\providecommand{\BIBentryALTinterwordspacing}{\spaceskip=\fontdimen2\font plus
\BIBentryALTinterwordstretchfactor\fontdimen3\font minus
  \fontdimen4\font\relax}
\providecommand{\BIBforeignlanguage}[2]{{%
\expandafter\ifx\csname l@#1\endcsname\relax
\typeout{** WARNING: IEEEtran.bst: No hyphenation pattern has been}%
\typeout{** loaded for the language `#1'. Using the pattern for}%
\typeout{** the default language instead.}%
\else
\language=\csname l@#1\endcsname
\fi
#2}}
\providecommand{\BIBdecl}{\relax}
\BIBdecl

\bibitem{Doriansw47:online}
``Dorian's wrath costly, destructive to {NC} roads, highways,''
  \url{https://www.wfmynews2.com/article/news/local/hurricane-dorian-damage-north-carolina-roads-ncdot/83-57c7235b-8903-48de-afe7-0a039734db50},
  (Accessed on 08/27/2020).

\bibitem{arrighi2019preparedness}
C.~Arrighi, M.~Pregnolato, R.~Dawson, and F.~Castelli, ``Preparedness against
  mobility disruption by floods,'' \emph{Science of the Total Environment},
  vol. 654, pp. 1010--1022, 2019.

\bibitem{coffrin2011strategic}
C.~Coffrin, P.~Van~Hentenryck, and R.~Bent, ``Strategic stockpiling of power
  system supplies for disaster recovery,'' in \emph{2011 IEEE Power and Energy
  Society General Meeting}.\hskip 1em plus 0.5em minus 0.4em\relax IEEE, 2011,
  pp. 1--8.

\bibitem{iloglu2020maximal}
S.~Iloglu and L.~A. Albert, ``A maximal multiple coverage and network
  restoration problem for disaster recovery,'' \emph{Operations Research
  Perspectives}, vol.~7, p. 100132, 2020.

\bibitem{orabi2010optimizing}
W.~Orabi, A.~B. Senouci, K.~El-Rayes, and H.~Al-Derham, ``Optimizing resource
  utilization during the recovery of civil infrastructure systems,''
  \emph{Journal of management in engineering}, vol.~26, no.~4, pp. 237--246,
  2010.

\bibitem{gaspard2007impact}
K.~Gaspard, M.~Martinez, Z.~Zhang, and Z.~Wu, ``Impact of {H}urricane {K}atrina
  on roadways in the {N}ew {O}rleans area,'' Technical Assistance Rep. No.
  07-2TA, LTRC Pavement Research Group, 2007.

\bibitem{chisolm2012impact}
E.~I. Chisolm and J.~C. Matthews, ``Impact of hurricanes and flooding on buried
  infrastructure,'' \emph{Leadership and Management in Engineering}, vol.~12,
  no.~3, pp. 151--156, 2012.

\bibitem{gehlot2019optimal}
H.~Gehlot, S.~Sundaram, and S.~V. Ukkusuri, ``Optimal sequencing policies for
  recovery of physical infrastructure after disasters,'' in \emph{2019 American
  Control Conference (ACC)}.\hskip 1em plus 0.5em minus 0.4em\relax IEEE, 2019,
  pp. 3605--3610.

\bibitem{hgehlot_tac}
------, ``Optimal policies for recovery of multiple systems after
  disruptions,'' \emph{arXiv preprint arXiv:1904.11615}, 2020.

\bibitem{janssen2010advances}
M.~Janssen, J.~Lee, N.~Bharosa, and A.~Cresswell, ``Advances in multi-agency
  disaster management: Key elements in disaster research,'' \emph{Information
  Systems Frontiers}, vol.~12, no.~1, pp. 1--7, 2010.

\bibitem{comerio2006estimating}
M.~C. Comerio, ``Estimating downtime in loss modeling,'' \emph{Earthquake
  Spectra}, vol.~22, no.~2, pp. 349--365, 2006.

\bibitem{godschalk2009estimating}
D.~R. Godschalk, A.~Rose, E.~Mittler, K.~Porter, and C.~T. West, ``Estimating
  the value of foresight: aggregate analysis of natural hazard mitigation
  benefits and costs,'' \emph{Journal of Environmental Planning and
  Management}, vol.~52, no.~6, pp. 739--756, 2009.

\bibitem{sun2020post}
J.~Sun and Z.~Zhang, ``A post-disaster resource allocation framework for
  improving resilience of interdependent infrastructure networks,''
  \emph{Transportation Research Part D: Transport and Environment}, vol.~85, p.
  102455, 2020.

\bibitem{lee2008multi}
W.-C. Lee and C.-C. Wu, ``Multi-machine scheduling with deteriorating jobs and
  scheduled maintenance,'' \emph{Applied Mathematical Modelling}, vol.~32,
  no.~3, pp. 362--373, 2008.

\bibitem{hsieh1997scheduling}
Y.-C. Hsieh and D.~L. Bricker, ``Scheduling linearly deteriorating jobs on
  multiple machines,'' \emph{Computers \& industrial engineering}, vol.~32,
  no.~4, pp. 727--734, 1997.

\bibitem{zhang2009schedulability}
F.~Zhang and A.~Burns, ``Schedulability analysis for real-time systems with
  {EDF} scheduling,'' \emph{IEEE Transactions on Computers}, no.~9, pp.
  1250--1258, 2009.

\bibitem{argon2008scheduling}
N.~T. Argon, S.~Ziya, and R.~Righter, ``Scheduling impatient jobs in a clearing
  system with insights on patient triage in mass casualty incidents,''
  \emph{Probability in the Engineering and Informational Sciences}, vol.~22,
  no.~3, pp. 301--332, 2008.

\bibitem{nilsson2017class}
P.~Nilsson and N.~Ozay, ``On a class of maximal invariance inducing control
  strategies for large collections of switched systems,'' in \emph{Proceedings
  of the 20th International Conference on Hybrid Systems: Computation and
  Control}.\hskip 1em plus 0.5em minus 0.4em\relax ACM, 2017, pp. 187--196.

\bibitem{nghiem2011green}
T.~X. Nghiem, M.~Behl, R.~Mangharam, and G.~J. Pappas, ``Green scheduling of
  control systems for peak demand reduction,'' in \emph{50th IEEE Conference on
  Decision and Control, and European Control Conference (CDC-ECC)}, 2011, pp.
  5131--5136.

\bibitem{cassandras2012optimal}
C.~G. Cassandras, X.~Lin, and X.~Ding, ``An optimal control approach to the
  multi-agent persistent monitoring problem,'' \emph{IEEE Transactions on
  Automatic Control}, vol.~58, no.~4, pp. 947--961, 2012.

\bibitem{lin2014optimal}
X.~Lin and C.~G. Cassandras, ``An optimal control approach to the multi-agent
  persistent monitoring problem in two-dimensional spaces,'' \emph{IEEE
  Transactions on Automatic Control}, vol.~60, no.~6, pp. 1659--1664, 2014.

\bibitem{eryilmaz2007fair}
A.~Eryilmaz and R.~Srikant, ``Fair resource allocation in wireless networks
  using queue-length-based scheduling and congestion control,'' \emph{IEEE/ACM
  Transactions on Networking (TON)}, vol.~15, no.~6, pp. 1333--1344, 2007.

\bibitem{buche2004control}
R.~Buche and H.~J. Kushner, ``Control of mobile communication systems with
  time-varying channels via stability methods,'' \emph{IEEE Transactions on
  Automatic Control}, vol.~49, no.~11, pp. 1954--1962, 2004.

\bibitem{righter1996optimal}
R.~Righter, ``Optimal policies for scheduling repairs and allocating
  heterogeneous servers,'' \emph{Journal of applied probability}, pp. 536--547,
  1996.

\end{thebibliography}
\end{document}